\documentclass[pra,twocolumn,showpacs,10pt,floatfix]{revtex4-1}

\usepackage{amsmath}
\usepackage{amsfonts}
\usepackage{amssymb}
\usepackage{mathrsfs}
\usepackage{latexsym}
\usepackage{graphicx}
\usepackage{color}

\newtheorem{theorem}{Theorem}

\newtheorem{definition}[theorem]{Definition}

\newenvironment{proof}[1][Proof]{\noindent\textbf{#1.} }{\ \rule{0.5em}{0.5em}}

\begin{document}

\title{Maxwell's demons in multipartite quantum correlated systems} 

\author{Helena C. \surname{Braga}}
\email{helenacbraga@gmail.com}
\author{Clodoaldo C. \surname{Rulli}}
\email{clodoaldorulli@gmail.com}
\author{Thiago R. de \surname{Oliveira}}
\email{tro@if.uff.br}
\author{Marcelo S. \surname{Sarandy}}
\email{msarandy@if.uff.br}

\affiliation{Instituto de F\'{\i}sica, Universidade Federal Fluminense, Av. Gal. Milton Tavares de Souza s/n, Gragoat\'a, 
24210-346, Niter\'oi, RJ, Brazil.}

\date{\today }

\begin{abstract}
We investigate the extraction of thermodynamic work by a Maxwell's demon in a multipartite quantum correlated 
system. We begin by adopting the standard model of a Maxwell's demon as a Turing machine, either in a classical 
or quantum setup  depending on its ability of implementing classical or quantum conditional dynamics, respectively. 
Then, for an $n$-partite system $(A_1, A_2, \cdots, A_n)$, we introduce a protocol of work extraction that 
bounds the advantage of the quantum demon over its classical counterpart through the amount of multipartite quantum 
correlation present in the system, as measured by a thermal version of the global quantum discord. This result  
is illustrated for an arbitrary $n$-partite pure state of qubits with Schmidt decomposition, where it is shown that 
the thermal global quantum discord exactly quantifies the quantum advantage. Moreover, we also consider the  
work extraction via mixed multipartite states, where examples of tight upper bounds can be obtained.  
\end{abstract}


\maketitle

%%%%%%%%%%%%%%%%%%%%%%%%%%
\section{Introduction}
%%%%%%%%%%%%%%%%%%%%%%%%%%

The concept of Maxwell's demon has introduced a deep relationship between information theory and 
thermodynamics~\cite{Maxwell:1871}. As originally proposed, the demon can be thought as a microscopic 
"intelligent" being capable of extracting work from a thermodynamic system at apparent no energy cost. 
As a simple example, consider a gas (initially in an equilibrium thermal state) contained in a chamber 
divided into two parts by an insulated wall. By direct inspection (followed by a post-selection) of 
fast particles, the demon would then be able to create a temperature gradient between the two parts, 
which could be used, e.g., as an energy resource for a thermal machine. Naturally, in order to provide a 
continuous work extraction, a cyclic process must be required, which imposes the erasure of the demon's memory 
for a complete thermodynamic accounting. Indeed, the irreversibility of the erasure operation, which is the main 
content of the Landauer principle~\cite{Landauer:61}, is the ultimate reason responsible for the conciliation of 
the Maxwell's demon with the second law of thermodynamics. 

In a modern perspective, we can take a Maxwell's demon as any device with the ability of information processing 
(as a computer modeled by a Turing Machine), where the extraction of work comes at the only cost of memory erasure 
at the end of the process. Remarkably, it has been shown by Zurek in Ref.~\cite{Zurek:03} that a quantum demon, 
which has the ability of implementing a quantum conditional dynamics through global operations over the system, 
can be more efficient in extracting work of a quantum system than any classical demon, which acts through local 
operations and classical communication to implement a classical dynamics. In a bipartite system-apparatus scenario, this difference can be 
quantified by the amount of quantum correlations between system and apparatus, as measured by a 
thermal version of the quantum discord (QD)~\cite{Ollivier:01,Zurek:03}. Indeed, QD has been identified as a general resource in quantum information 
protocols (see, e.g., Refs.~\cite{Madhok:12,Gu:12,locking,Roa:11}). In quantum computation, it has been 
conjectured as the origin of speed up in the deterministic quantum computation with one qubit (DQC1) mixed-state 
model~\cite{DQC1}. Moreover, remarkable applications of QD have also been found in the characterization of quantum 
phase transitions~\cite{QPT} and in the description of quantum dynamics under decoherence~\cite{Decoh}. In this 
context, an operational interpretation of the QD in terms of the efficiency of a Maxwell's demon establishes a 
solid framework to investigate its conceptual role in quantum thermodynamics as well as to inspire new QD-based 
quantum protocols. 

In this work, we aim at investigating the efficiency of both classical and quantum Maxwell's demons in the multipartite 
scenario. In particular, we are interested in analyzing the relationship between the quantum advantage and the existence of multipartite 
quantum correlations. More specifically, provided $n$ copies of an $n$-partite system $(A_1, A_2, \cdots, A_n)$, we introduce a protocol 
of work extraction defined through a sequence of intermediate steps, where the Maxwell demon (either classical or quantum) 
uses a subsystem $A_i$ of copy $i$ ($i=1,\cdots,n$) as a measurement apparatus at a each step, with $i$ a sequential label chosen 
at demon's will. The demon is also required to erase its memory at the end of the process, so that the thermodynamic 
accounting does not disregard the irreversible local cost of erasure. In this context, we will show that the advantage of the quantum demon 
over its classical counterpart in extracting work (from the $n$ copies) is bounded through the amount of multipartite quantum correlation, as measured by a 
%%v2: typo corrected: D_th removed here
thermal version of the global quantum discord (GQD)~\cite{Rulli:11}. 
%%%
Indeed, GQD has been introduced as a multipartite approach to quantify quantum correlations, which has been applied to the characterization of 
quantum phase transitions in many-body systems~\cite{Rulli:11,Campbell:11}. Moreover, it can be witnessed with no 
extremization procedure~\cite{Saguia:11} and it has been used as a tool to define a monogamy relationship for the 
standard QD~\cite{Braga:12}. Here, we provide a generalization of GQD to a thermodynamical scenario as well as an 
operational interpretation of this thermal version of GQD. We illustrate our procedure of work extraction 
for an arbitrary $n$-partite pure state of qubits with Schmidt decomposition, where it is shown that the thermal 
GQD exactly quantifies the quantum advantage of the Maxwell's demon. Moreover, we also provide examples of work 
extraction via mixed multipartite states. These examples allow for the discussion of the tightness of the 
bound in situations where saturation is not always achieved.

%%%%%%%%%%%%%%%%%%%%%%%%%%%%%%%%%%%%%%%%%%%%%%%%%%%%%%%%%%%%%%
\section{Maxwell's demon and bipartite quantum correlations}
%%%%%%%%%%%%%%%%%%%%%%%%%%%%%%%%%%%%%%%%%%%%%%%%%%%%%%%%%%%%%%

Work and information are equivalent operational concepts~~\cite{Szilard:1929} (see also, e.g.,  
Refs.~\cite{Alicki:04,Horodecki:05,Shikano:11} for more recent discussions). If a system ${\cal S}$ described by a $d$-level 
pure state $|\psi\rangle$ is available as a resource, we can extract work from ${\cal S}$ by letting it expand 
throughout Hilbert space while in contact with a thermal reservoir at temperature $T$. For such an isothermal process, 
one can draw work $W = k T \log d$ out of the heat bath, where $k$ is the Boltzmann constant adapted to deal with the 
entropy in bits (so that $\log \equiv \log_2$). For a mixed state $\rho$, less work is possible to be extracted, since 
less knowledge (information) about the state of the system is available. In this situation, we should discount the 
necessary work to be performed over the system to drive it to a pure state. This yields  
\begin{equation}
W = k T \left[ \log d - S(\rho) \right],
\label{work-rho}
\end{equation}
where $S(\rho) = -{\textrm{Tr}} \rho \log \rho$ is the von Neumann entropy associated with the quantum state $\rho$. 
In a bipartite system-apparatus ($\cal{SA}$) scheme, we can then write the work $W^Q$ extracted by a quantum demon as 
\begin{equation}
W^Q = k T \left[ \log d_{SA} - S(\rho_{SA}) \right],
\label{work-Q-rho}
\end{equation}
where $d_{SA}$ is the dimension of ${\cal SA}$ and $S(\rho_{SA})$ is its joint von Neumann entropy. 
A classical demon, on the other hand, first implements a local measurement $\{\Pi_{A}^{k}\}$ on
the apparatus, using the measured state to extract work $\log(d_A)-H(\{p_a\})$ from ${\cal A}$, where $H(\{p_a\})=-\sum_a p_a \log p_a$ 
is the Shannon entropy for the probability distribution $\{p_a\}$ associated with the local measurement $\{\Pi_{A}^{k}\}$. Then, an update of ${\cal S}$ 
is performed based on the outcome read from the apparatus and work $\log(d_S)-S(\rho_{S} | \{\Pi_{A}^{k}\})$ is extracted from ${\cal S}$, 
where $S(\rho_{S} | \{\Pi_{A}^{k}\})$ is the 
conditional entropy accessible through $\{\Pi_{A}^{k}\}$, which is given by the weighted average 
\begin{equation}
S(\rho_{S} | \{\Pi_{A}^{k}\} ) = \sum_i p_i S(\rho_i), 
\end{equation}
with $S(\rho_i)$ denoting the von Neumann entropy of the post-measurement state 
$\rho_i = (1/p_i) (I_S \otimes \Pi_A^i) \rho_{SA} (I_S \otimes \Pi_A^i)$ and 
$p_i = {\textrm{Tr} [(I_S \otimes \Pi_A^i) \rho_{SA} (I_S \otimes \Pi_A^i)] }$. 
In this work, for simplicity, 
we will typically restrict $\Pi_{A}^{k}$ as
%%v2:
rank-one 
%%%
orthogonal projective measurements rather than arbitrary positive 
operator-valued measures (POVMs).   
The total amount of work extracted is then given by 
\begin{equation}
W^C = k T \left[ \log d_{SA} -  S_A (\rho_{SA}) \right],
\label{work-C-rho}
\end{equation}
where $S_A (\rho_{SA})$ is the locally accessible joint entropy, which reads
\begin{equation}
S_A (\rho_{SA}) =  H(\{p_a\}) + S(\rho_{S} | \{\Pi_{A}^{k}\}) .
\label{LAjoint}
\end{equation}
Remarkable, the minimum difference between $W^Q$ and $W^C$, which is given by 
the best classical strategy, can be quantified by the quantum correlation 
between ${\cal S}$ and ${\cal A}$, as measured by the thermal QD. As defined in Ref.~\cite{Zurek:03}, 
the thermal QD $\mathcal{D}_{th}\left( S | A \right)$ for a composite system $\cal{SA}$ can be suitably expressed 
(with respect to $\cal{A}$) as the difference between the quantum mutual information 
\begin{equation}
I(\rho_{SA}) = S(\rho_S) + S(\rho_A) - S(\rho_{SA})
\label{I-bipart}
\end{equation}
and the locally accessible mutual information 
\begin{equation}
J_A(\rho_{SA}) = S(\rho_S) + S(\rho_A) - S_A(\rho_{SA}), 
\end{equation}
with the difference $I(\rho_{SA}) - J_A(\rho_{SA})$ minimized over all local measurements $\{\Pi_{A}^{k}\}$. 
This reads
\begin{equation}
\mathcal{D}_{th}\left( S | A \right) = \min_{\{\Pi_{A}^{k}\}} \left[ S_A(\rho_{SA}) - S(\rho_{SA}) \right].
\label{QD-asym}
\end{equation}
Then, by using Eqs.~(\ref{work-Q-rho}),~(\ref{work-C-rho}) 
and~(\ref{QD-asym}), it has been shown in Ref.~\cite{Zurek:03} that
\begin{equation}
\Delta W \equiv \min_{i} (W^Q - W^{C_i}) =  kT \, \mathcal{D}_{th}\left( S | A \right) ,
\label{deltaW}
\end{equation}
where $\min_{i}$ denotes the minimum over the difference $ (W^Q - W^{C_i})$ for all the possible strategies 
$\{\Pi_{A}^{k}\}$ for work extraction adopted by a classical demon $C_i$. In terms of conditional entropies, the 
thermal QD can be also written as 
\begin{eqnarray}
\mathcal{D}_{th}\left( S | A \right) &=& \min_{\{\Pi_{A}^{k}\}} \left\{ 
\left[H(\{p_a\}) + S(\rho_{S} | \{\Pi_{A}^{k}\})\right] \right. \nonumber \\
&& \hspace{0.7cm} \left. - \left[S(\rho_A) + S(\rho_S | \rho_A) \right]\right\},
\label{QD-asym-ce}
\end{eqnarray}
with 
\begin{equation}
S(\rho_S | \rho_A) = S(\rho_{SA}) - S(\rho_A)
\end{equation}
denoting the entropy of $\cal{S}$ conditional on $\cal{A}$. 
Therefore, the thermal QD $\mathcal{D}_{th}(S|A)$ is distinct of the original QD  $\mathcal{D}(S|A)$ proposed in 
Ref.~\cite{Ollivier:01}, which is given by 
\begin{equation}
\mathcal{D} \left( S | A \right) = \min_{\{\Pi_{A}^{k}\}} \left[ S(\rho_{S} | \{\Pi_{A}^{k}\}) - S(\rho_S | \rho_A) \right].
\label{QD-asym-orig}
\end{equation}
In particular, the thermal QD is also referred as the one-way work deficit~\cite{Modi:12}. Moreover, it follows that 
$\mathcal{D}_{th} \left( S | A \right) \ge \mathcal{D} \left( S | A \right)$~\cite{Terno:10}. 

%%%%%%%%%%%%%%%%%%%%%%%%%%%%%%%%%%%%%%%%%%%%%%%%%%%%%%%%%%%%%%%%%%%%%%%%%%%%%
\section{Maxwell's demon and multipartite quantum correlations}
%%%%%%%%%%%%%%%%%%%%%%%%%%%%%%%%%%%%%%%%%%%%%%%%%%%%%%%%%%%%%%%%%%%%%%%%%%%%%

Let us now present a thermodynamic protocol able to provide an operational interpretation for multipartite 
quantum correlations as measured by the thermal GQD. In order to introduce the thermal GQD, let us first rewrite 
$\mathcal{D}_{th}\left( S | A \right)$ in terms of loss of total correlation after a non-selective 
measurement~\cite{Luo:10}. This is a measurement characterized by an unrevealed outcome, i.e. the system is measured 
but the outcome is not read out. In a bipartite system $\cal{SA}$ composed of subsystems $\cal{S}$ and $\cal{A}$, 
the thermal QD as given by Eq.~(\ref{QD-asym-ce}) can be expressed as
\begin{eqnarray}
\mathcal{D}_{th}\left( S | A \right) &=& \min_{\{\Pi_{A}^{k}\}} \left\{ 
\left[ I(\rho_{SA}) - I(\Phi_{A}\left( \rho_{SA}\right)) \right] \right. \nonumber \\
&&\left. \hspace{0.7cm} + \left[ H(\{p_a\}) - S(\rho_A) \right] \right\},
\label{QD-asym-mi}
\end{eqnarray}
where $\Phi _{A}\left( \rho_{SA}\right)$ denotes a non-selective measurement $\{\Pi_{A}^{j}\}$ on part $\cal{A}$ of 
$\rho_{SA}$, which reads 
$\Phi _{A}\left( \rho_{SA}\right) = \sum_{j} \left( I_{S} \otimes \Pi_{A}^{j} \right) \rho_{SA} 
\left( I_{S}\otimes \Pi_{A}^{j}\right)$. In order to derive Eq.~(\ref{QD-asym-mi}), we have used that 
$S(\rho_{S} | \{\Pi_{A}^{k}\}) - S(\rho_S | \rho_A) = I(\rho_{SA}) - I(\Phi_{A}\left( \rho_{SA}\right))$~\cite{Rulli:11}. 
Note that Eq.~(\ref{QD-asym-mi}) is asymmetric with respect to measurement on 
$\cal{S}$ and $\cal{A}$, which reflects an asymmetry in the roles of system $\cal{S}$ and apparatus $\cal{A}$. 
Due to the asymmetry of QD, a strictly classical bipartite state requires both $\mathcal{D}_{th}\left( S | A \right) = 0$ and 
$\mathcal{D}_{th}\left( A | S \right) = 0$. Indeed, this corresponds to a density operator 
${\rho}_{AB} = \sum_{i,j} p_{ij} |i\rangle\langle i| \otimes |j\rangle\langle j|$, where $p_{ij}$ is a joint 
probability distribution and the sets $\{|i\rangle\}$ and $\{|j\rangle\}$ constitute orthonormal bases for 
the systems $A$ and $B$, respectively. For an arbitrary bipartite state composed by subsystems $A_1$ and $A_2$, 
such strictly classical states can also be identified by a single measure, which is the symmetrized version of QD
\begin{eqnarray}
\mathcal{D}_{th}\left( A_1 : A_2 \right) &=& 
\min_{\Phi} \left\{\frac{}{}
I({\rho}_{A_1 A_2}) - I(\Phi _{A_1 A_2}\left( {\rho}_{A_1 A_2}\right))\right. \nonumber \\
&&\left. \hspace{0.7cm} + \sum_{i=1}^{2} \left[ H(\{p_{a_i}\}) - S(\rho_{A_i}) \right] \right\}, 
\label{BipartiteDiscord}
\end{eqnarray}
where the measurement operator $\Phi _{A_1 A_2}$ is given by
\begin{equation}
\Phi _{A_1 A_2}\left( {\rho}_{A_1 A_2}\right) =\sum_{j,k} \left({\Pi}_{A_1}^{j}\otimes 
{\Pi}_{A_2}^{k} \right) {\rho}_{A_1 A_2} \left({\Pi}_{A_1}^{j}\otimes {\Pi}_{A_2}^{k}\right) \, .
\end{equation}
By explicitly using Eq.~(\ref{I-bipart}) and the fact that $H(\{p_{a_i}\})=S(\Phi_{A_i}(\rho_{A_i}))$ we can rewrite 
Eq.~(\ref{BipartiteDiscord}) as
\begin{equation}
\mathcal{D}_{th}\left( A_1 : A_2 \right) = \min_{\Phi} 
\left[ S(\Phi _{A_1 A_2}\left( {\rho}_{A_1 A_2}\right)) - S({\rho}_{A_1 A_2}) \right]
\label{Bipartite-QD-sym-final}
\end{equation}
Eq.~(\ref{Bipartite-QD-sym-final}) provides the thermal generalization of the symmetric QD considered in 
Ref.~\cite{Maziero:10} and experimentally witnessed in Refs.~\cite{Auccaise:11,Aguilar:12}. 
The vanishing of $\mathcal{D}_{th}\left( A_1 : A_2 \right)$ 
occurs if and only if the state is fully classical. In particular, the absence of $\mathcal{D}_{th}\left( A_1 : A_2 \right)$ 
can be taken as the key ingredient for local sharing of pre-established correlations (local broadcasting)~\cite{Piani:08}. 

Generalizations of quantum discord to multipartite states have been considered in different 
scenarios~\cite{Modi:11,Chakrabarty:10,Okrasa:11,Rulli:11,Celeri:11,Xu:12}, which intend to account for quantum 
correlations that may exist beyond pairwise subsystems in a composite system. In this direction, one possible 
approach to account multipartite quantum correlations is to start from the symmetrized QD and then to systematically 
extend it to the multipartite scenario. This originates GQD as a measure of global quantum discord, as proposed in 
Ref.~\cite{Rulli:11}. GQD is symmetric with respect to subsystem exchange and shown to be non-negative for arbitrary states~\cite{Rulli:11}. Moreover, it can be detected through a convenient (with no extremization procedure) witness operator~\cite{Saguia:11}. 
In terms of operational interpretation, GQD may play a role in quantum communication, in the sense that its absence 
means that the quantum state simply describes a classical probability multidistribution 
$ \sum_{i_1,\cdots,i_n} p_{i_1 \cdots i_n} |i_1\rangle\langle i_1| \otimes \cdots \otimes |i_n\rangle\langle i_n|$ 
(with $p_{i_1 \cdots i_n} \ge 0$, $\sum p_{i_1 \cdots i_n} = 1$) and, therefore, allows for local broadcasting~\cite{Piani:08}. 
Here we will propose a slightly distinct version of GQD, which will be motivated by an operational interpretation in terms of work extraction in quantum thermodynamics. We will refer to this multipartite measure of quantum correlation as thermal GQD, whose definition is given below. 

\begin{definition}
The thermal GQD $\mathcal{D}_{th}\left( A_1 : \cdots :A_n \right)$ for an arbitrary multipartite state ${\rho}$ composed of subsystems $A_1, \cdots, A_n$ is defined as

\begin{equation}
\mathcal{D}_{th}\left( A_1 : \cdots : A_n \right) = \min_{\Phi} \left[\,
S(\Phi_{A_1 \cdots A_n}\left( {\rho} \right)) - S({\rho}) \right], 
\label{gqd-def}
\end{equation}
where 

\begin{equation}
\Phi_{A_1 \cdots A_n}\left( {\rho} \right) = \sum_{k} {{\Pi}}_{k} \, {\rho} \, 
{{\Pi}}_{k},
\label{ns-measu}
\end{equation}
with ${\Pi}_{k} = {\Pi}_{A_1}^{j_1} \otimes \cdots \otimes {\Pi}_{A_n}^{j_n}$ denoting a set of local measurements and 
$k$ an index string $(j_1 \cdots j_n$). 
\end{definition}
We will now show that the thermal GQD provides an upper bound for the sum of a sequence of bipartite asymmetric thermal discords, 
which will imply in the interpretation of the thermal GQD in terms of a limit of work extraction through 
a protocol of local operations in a multipartite system. This is provided by the Theorem below.

\begin{theorem}
The thermal GQD $\mathcal{D}_{th}\left( A_1 : \cdots :A_n \right)$ for an arbitrary multipartite state 
${\rho}$ composed of subsystems $A_1, \cdots, A_n$ satisfies the inequality
\begin{equation}
\mathcal{D}_{th}\left( A_1 : \cdots : A_n \right) \ge  \sum_{i=1}^{n} \min_{\Phi_{A_1 \cdots A_{i-1}}} {\cal D}_{th}(\Phi_{A_1 \cdots A_{i-1}}(\rho) | A_{i}) , 
\label{Dge-proof} 
\end{equation}
where the asymmetric bipartite contributions ${\cal D}_{th}(\Phi_{A_1 \cdots A_{i-1}}(\rho) | A_{i})$ $(\forall i)$ are provided by Eq~(\ref{QD-asym-mi}), 
with $\min_{\Phi_{A_1 \cdots A_{0}}} {\cal D}_{th}(\Phi_{A_1 \cdots A_{0}}(\rho) | A_{1}) \equiv {\cal D}_{th}(\rho | A_{1})$.
\end{theorem}

\begin{proof}
In Eq.~(\ref{gqd-def}), let us consider the difference of joint entropies for a fixed measurement 
$\Phi_{A_1 \cdots A_n} \left( {\rho} \right)$, which yields 
\begin{equation}
\mathcal{D}_\Phi\left( A_1 : \cdots : A_n \right) = S(\Phi_{A_1 \cdots A_n}\left( {\rho} \right)) - S({\rho}) .
\end{equation}
By rewriting $\mathcal{D}_\Phi\left( A_1 : \cdots : A_n \right)$ in terms of the multipartite mutual information, we obtain
\begin{eqnarray}
\mathcal{D}_{\Phi}\left( A_1 : \cdots : A_n \right) &=& 
\left\{\frac{}{} I({\rho}) - I(\Phi_{A_1 \cdots A_n}\left( {\rho}\right))\right. \nonumber \\
&&\left. \hspace{-0.3cm} + \sum_{i=1}^{n} \left[ H(\{p_{a_i}\}) - S(\rho_{A_i}) \right] \right\}, 
\label{ThermalGQD-mi}
\end{eqnarray}
where $I({\rho})$ and $I(\Phi_{A_1 \cdots A_n}\left( {\rho} \right))$ are generalizations of the 
mutual information to the multipartite setting~\cite{Groisman:05}, which are given by
\begin{eqnarray}
\hspace{-0.4cm}I({\rho}_{A_1 \cdots A_n}) &=& \sum_{k=1}^{n} S\left( {\rho}_{A_k}\right) - 
S\left( {\rho}_{A_1 \cdots A_n} \right) , \label{Imulti1} \\
\hspace{-0.5cm}I(\Phi_{A_1 \cdots A_n}\left( {\rho} \right)) &=& \sum_{k=1}^{n} S\left( \Phi \left( {\rho}_{A_k}\right) \right) 
- S\left(\Phi_{A_1 \cdots A_n} \left( {\rho}\right) \right) ,
\label{Imulti2}
\end{eqnarray}
where 
\begin{equation}
\Phi \left( {\rho}_{A_k} \right) = \sum_{k^\prime} {\Pi}_{A_k}^{k^\prime} \, {\rho}_{A_k} \,
{\Pi}_{A_k}^{k^\prime},
\end{equation}
with ${\rho}_{A_k}$ denoting the marginal density operator for subsystem $A_k$. In Eq.~(\ref{ThermalGQD-mi}), 
we now rearrange the terms by adding and subtracting the contributions 
 $I(\Phi_{A_1 \cdots A_{i}}(\rho))$ for all $i \in \{1,\cdots,n-1\}$.  
We then obtain
\begin{equation}
{\mathcal{D}}_\Phi(A_1:\cdots:A_n) =  \sum_{i=1}^{n} {\cal D}_\Phi(\Phi_{A_1 \cdots A_{i-1}}(\rho) | A_{i}),
\label{interm-proof}
\end{equation}
where
\[
\mathcal{D}_{\Phi}\left( S | A \right) =  
I(\rho_{SA}) - I(\Phi_{A}\left( \rho_{SA}\right)) 
+  H(\{p_a\}) - S(\rho_A) .
\]
We can relate Eq.~(\ref{interm-proof}) to the thermal GQD through
\begin{equation}
\mathcal{D}_{th}\left( A_1 : \cdots : A_n \right) = \min_{\Phi} {\mathcal{D}}_\Phi(A_1:\cdots:A_n).
\end{equation}
This yields
\[
\mathcal{D}_{th}\left( A_1 : \cdots : A_n \right) \ge 
\sum_{i=1}^{n} \min_{\Phi} {\cal D}_\Phi(\Phi_{A_1 \cdots A_{i-1}}(\rho) | A_{i}),
\]
which implies in Eq.~(\ref{Dge-proof}). 
\end{proof}

Remarkably, Theorem~2 provides a relationship between a symmetric measure of quantum correlation (GQD) 
and a composition of asymmetric operations ($ {\cal D}_{th}(\Phi_{A_1 \cdots A_{i-1}}(\rho) | A_{i})$), which 
involve sequential local measurements over distinct subsystems. In particular, the local measurements yield a 
sequence of bipartite discords that are chained following a rather simple rule. 
As an illustration, for a bipartite system ($n=2$) composed of subsystems $A$ and $B$, 
we can write
\begin{equation}
\mathcal{D}_{th}\left( A : B \right) \ge  {\cal D}_{th}(\rho | A) + 
\min_{\Phi_{A}} {\cal D}_{th}(\Phi_{A}(\rho) | B), 
\end{equation}
while for a tripartite system ($n=3$) composed by subsystems $A$, $B$, and $C$, the bound assumes the form
\begin{eqnarray}
\mathcal{D}_{th}\left( A : B : C \right) &\ge &  {\cal D}_{th}(\rho | A) 
+ \min_{\Phi_{A}} {\cal D}_{th}(\Phi_{A}(\rho) | B) \nonumber \\
&& + \min_{\Phi_{AB}} {\cal D}_{th}(\Phi_{AB}(\rho) | C). 
\label{tripartite-bound}
\end{eqnarray}
Moreover, Eq.~(\ref{Dge-proof}) ensures as a by-product that 
$\mathcal{D}_{th}\left( A_1 : \cdots : A_n \right)\ge 0$. 
Note also that the derivation of Theorem~2 also 
allows for other less restricted bounds. For instance, suppose we take $\Phi_{A_1 \cdots A_n}(\rho)$ 
in Eq.~(\ref{interm-proof}) with local measurement operators ${\Pi}_{k}$ defined by the eigenprojectors of the reduced 
states $\rho_{A_k}$, with $k=1,\cdots,n$. This corresponds to the measurement-induced disturbance (MID) basis~\cite{Luo:08}. 
For this specific basis, Eq.~(\ref{interm-proof}) reads
\[
{\mathcal{D}}_{MID}(A_1:\cdots:A_n) =  \sum_{i=1}^{n} {\cal D}_{MID}(\Phi_{A_1 \cdots A_{i-1}}(\rho) | A_{i}).
\]
By using that ${\cal D}_{MID}(S|A) \ge {\cal D}_{th}(S|A)$ for an arbitrary state $\rho$~\cite{Terno:10}, we can 
establish
\begin{equation}
{\mathcal{D}}_{MID}(A_1:\cdots:A_n) \ge \sum_{i=1}^{n} \min_{\Phi_{A_1 \cdots A_{i-1}}}  \hspace{-0.35cm}{\cal D}_{th}(\Phi_{A_1 \cdots A_{i-1}}(\rho) | A_{i}).
\label{MID-th}
\end{equation}
Eq.~(\ref{MID-th}) provides an upper bound  that is easier to compute than Eq.~(\ref{Dge-proof}), since no extremization is required to determine 
${\mathcal{D}}_{MID}(A_1:\cdots:A_n)$. However, 
as we will see, it can be less tight than ${\mathcal{D}}_{th}(A_1:\cdots:A_n)$. 

%%v2:
%%%%%%%%%%%%%%%%%%%%%%%%%%%%%%%%%%%%%%%%%%%%%%%%%%%%%%%%%%
\section{Thermodynamic interpretation of the thermal GQD}
%%%%%%%%%%%%%%%%%%%%%%%%%%%%%%%%%%%%%%%%%%%%%%%%%%%%%%%%%%
The decomposition of the thermal GQD as provided by Eq.~(\ref{Dge-proof}) allows for a thermodynamic interpretation of the multipartite correlations in terms 
of work extraction by Maxwell's demons. 
%%v2
In particular, Eq.~(\ref{Dge-proof}) implies that the thermal GQD is an upper bound for a series of 
bipartite thermal QDs, which are individually related to differences of performances between 
quantum and classical demons. These bipartite QDs involve locally measured states, 
which can be associated with a sequence of work extractions in a multi-copy version 
of the multipartite system.
In this direction, we will consider, as a physical resource, $n$ copies of an $n$-partite quantum system $(A_1, A_2, \cdots, A_n)$. 
The extraction of work from the system by either a classical or quantum Maxwell's demon will generalize the bipartite protocol through the following procedure: 

\vspace{0.35cm}

\noindent $\bullet$ \hspace{0.1cm} The demon {\it sequentially} takes a subsystem $A_i$ of copy $i$ as a measurement apparatus ($i=1,\cdots,n$). 

\vspace{0.35cm}

\noindent $\bullet$ \hspace{0.1cm}  The quantum demon can then apply conditional global quantum operations by using the apparatus as a control 
system. For the classical demon, arbitrary local measurements over the apparatus are allowed. 

\vspace{0.35cm}

\noindent $\bullet$ \hspace{0.1cm}  For all the subsystems $A_j$, with $j < i$ (in the copy $i$), the classical demon also realizes a non-selective measurement 
such as in Eq.~(\ref{ns-measu}) (with no memory cost) in such a way to minimize the difference of its local extracted 
work at subsystem $A_i$ with respect to the global work extracted by the quantum demon. For the subsystem $A_i$, a selective local measurement is performed 
(with memory cost), followed by the effective work extraction from the copy $i$ of the system. A schematic view of this procedure is provided in Fig.~\ref{f1}. 

\vspace{0.4cm}

\begin{figure}[th]
\centering {\includegraphics[angle=0,scale=0.31]{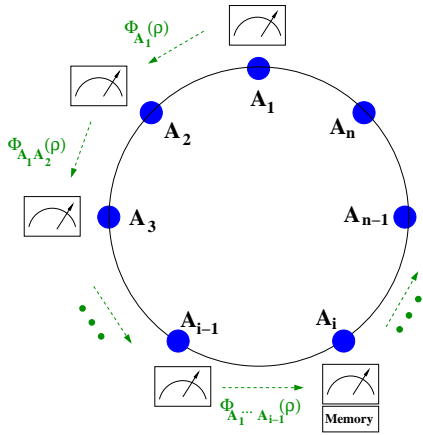}}
\caption{(Color online) Protocol for a sequential work extraction in a multipartite scenario, where the subsystem 
$A_i$ is taken as an apparatus at the intermediate step $i$. For the subsystems $A_j$, with $j < i$, the classical demon applies 
a non-selective measurement to optimize the local extracted work at subsystem $A_i$.}
\label{f1}
\end{figure}

Here we will be interested in the
total amount of work that is possible to be extracted by adding the partial work contributions. 
By extracting work under this procedure, we can then establish an upper bound for the advantage of 
the quantum demon with respect to the classical demon in terms of the thermal GQD. We begin by taking 
the system in a general $n$-partite state $\rho$ and $A_1$ as the  
apparatus. Then, by using Eq.~(\ref{deltaW}), we obtain that the difference of work between a quantum demon 
and the most efficient classical demon is 
$\Delta W_{1} = k\,T\,{\cal D}_{th}(\rho | A_1)$. For the next step, both classical and quantum demons take the second
copy of the system and use $A_2$ as the apparatus. The classical demon is also allowed to perform a non-selective 
measurement over $A_1$ in such a way to minimize the difference with respect to the quantum demon concerning the work 
extraction from $A_2$. Then, we will have 
$\Delta W_{2} = k\,T\, \min_{\Phi_{A_1}} {\cal D}_{th}(\Phi_{A_1}(\rho) | A_2)$.  After $n$ steps, we denote the total work difference as 
$\Delta W_t \equiv \sum_i \Delta W_{i}$. This yields
\begin{equation}
\frac{\Delta W_t} {k\,T\,} = \sum_{i=1}^{n} \min_{\Phi_{A_1 \cdots A_{i-1}}}{\cal D}_{th}(\Phi_{A_1 \cdots A_{i-1}}(\rho) | A_{i}).
\label{DW-proof}
\end{equation}
Hence, by inserting Eq.~(\ref{DW-proof}) into Eq.~(\ref{Dge-proof}), we obtain that GQD provides an upper bound for the 
difference $\Delta W_t$ of total work between a quantum and a classical Maxwell's demon, i.e.
\begin{equation}
\Delta W_t \le k\,T \, {\cal D}_{th}\left( A_1 : \cdots : A_n \right) .
\label{DN}
\end{equation}
As discussed above, an upper (but less strict) bound can be also established in terms of the MID basis, i.e. 
$\Delta W_t \le k\,T \, {\cal D}_{MID}\left( A_1 : \cdots : A_n \right)$. In both cases, note that the invariance 
of the upper bound under exchange of subsystems keeps it robust to a change in the order of measurements for the 
subsystems $A_i$. Moreover, as we will show, we can illustrate the applicability 
of Eq.~(\ref{DN}) in situations where the bound is rather tight or even is saturated. 

%%%%%%%%%%%%%%%%%%%%%%%%%%
\section{Illustrations}
%%%%%%%%%%%%%%%%%%%%%%%%%%

%%%%%%%%%%%%%%%%%%%%%%%%%%
\subsection{Pure states with Schmidt decomposition}
%%%%%%%%%%%%%%%%%%%%%%%%%%

Let us illustrate the upper bound given by Eq.~(\ref{DN}) for the quantum advantage of the Maxwell's demon  
in the case of multipartite pure states $|\psi\rangle$ that admit Schmidt decomposition, whose explicit conditions of 
existence are discussed in Ref.~\cite{SchDec}. We will assume that the system is composed by a set of qubits. 
In such a case, we can write 
$|\psi\rangle = \sum_{i=1}^{2} \sqrt{p_i} |i_{A_1}\rangle \otimes \cdots \otimes |i_{A_n}\rangle$, where 
$\{|i_{A_k}\rangle\}$ are orthonormal bases, $p_i \ge 0$, and $\sum_i p_i = 1$.
For the density operator ${\rho}_{A_1 \cdots A_n}=|\psi\rangle \langle \psi |$, we obtain
\begin{equation}
{\rho}_{A_1 \cdots A_n} = \sum_{i,j=1}^2 \sqrt{p_i p_j} |i_{A_1} \cdots i_{A_n}\rangle \langle j_{A_1} \cdots j_{A_n}| .
\end{equation}
Since Schmidt decomposition implies equal spectrum for all single-qubit reduced density operators ${\rho}_{A_k}$, we obtain 
that $S({\rho}_{A_k}) = -\sum_{i=1}^2 p_k \log_2 p_k \equiv S $, for any individual subsystem $A_k$. Therefore, 
the mutual information is $I({\rho}_{A_1 \cdots A_n}) = n \, S$. In order to consider measurements 
$\Phi\left( {\rho}_{A_1 \cdots A_n} \right)$ over ${\rho}_{A_1 \cdots A_n}$, it can be shown that, 
by adopting projective (von Neumann) measurements, the minimization of the loss of correlation is obtained in Schmidt 
basis, namely, $\{{\Pi}_{A_k}^{i}\} = \{|i_{A_k}\rangle \langle i_{A_k}|\}$. This is a consequence of both the 
group homomorphism of $U(2)$ to $SO(3)$ and the monotonicity of entropy under majorization (see discussion for 
the state $(|0 \cdots 0 \rangle + |1 \cdots 1 \rangle)/\sqrt{2}$ in Ref.~\cite{Xu:12}). 
Then, $\Phi\left( {\rho}_{A_k} \right) = {\rho}_{A_k}$, which implies 
$S(\Phi\left( {\rho}_{A_k} \right)) = S$. Moreover
$\Phi({\rho}_{A_1 \cdots A_n}) = \sum_{i=1}^2 p_i |i_{A_1} \cdots i_{A_N}\rangle \langle i_{A_1} \cdots i_{A_n}|$.
Therefore, the mutual information after measurement is $I(\Phi({\rho}_{A_1 \cdots A_n})) = (n-1) \, S$ and the Shannon entropy $H(\{p_{A_k}\})$ for the 
subsystem $A_k$ for a measurement in Schmidt basis obeys $H(\{p_{A_k}\})=S(\rho_{A_k})$. This  yields $\mathcal{D}_{th}\left( A_1 : \cdots :A_n \right) = S$. 

As an example, let us consider an $n$-qubit pure state 
$|\psi^{(i)}_{A_1 \cdots A_n} \rangle = \alpha |0 \cdots 0 \rangle + \beta |1 \cdots 1 \rangle$, with 
$|\alpha|^2+|\beta|^2 = 1$. In this case, 
$\mathcal{D}_{th}\left( A_1 : \cdots :A_n \right) = S =  -|\alpha|^2 \log  |\alpha|^2 - |\beta|^2 \log  |\beta|^2$. 
The optimal extraction of work by quantum and classical demons are performed through the following strategy:  
By using a qubit state $|0\rangle_D$ as a memory and $A_1$ as an apparatus, 
the quantum demon is able to purify all the individual states of the subsystems $A_1$, $\cdots$, $A_n$, while reseting its memory 
to the ready-to-measure state $|0\rangle_D$. This is obtained through the quantum circuit exhibited in Fig.~\ref{f2}. 
More specifically, this circuit drives the initial state $|\psi^{(i)}_{A_1 \cdots A_n} \rangle \otimes |0\rangle_D$  
to the final state 
\begin{equation}
|\psi^{(f)}_{A_1 \cdots A_n} \rangle \otimes |0\rangle_D = (\alpha |0\rangle + \beta |1\rangle)_{A_1}  \otimes 
|0\rangle_{A_2} \cdots |0\rangle_{A_n} \otimes |0\rangle_{D} .
\end{equation}
Therefore, since any individual subsystem is in a pure state, we obtain from Eq.~(\ref{work-rho}) that the work $W^Q$ 
extracted by the quantum demon is 
\begin{equation}
W^Q = n \, k\, T \, \log 2 .
\label{work-Q-sd}
\end{equation}

\begin{figure}[th]
\centering {\includegraphics[angle=0,scale=0.4]{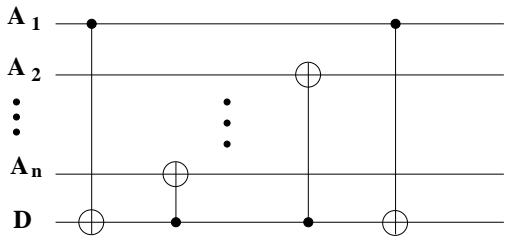}}
\caption{Quantum circuit for work extraction in the case of an $n$-qubit pure state $|\psi^{(i)}_{A_1 \cdots A_n} \rangle = \alpha |0 \cdots 0 \rangle + \beta |1 \cdots 1 \rangle$. 
The initial state is  $|\psi^{(i)}_{A_1 \cdots A_n} \rangle \otimes |0\rangle_D$. 
The quantum demon is able to purify the individual states of the subsystems $A_1$, $\cdots$, $A_n$, while reseting 
its memory to the ready-to-measure state $|0\rangle_D$.}
\label{f2}
\end{figure}

Concerning the classical demon, we assume that one bit of memory is available. Then, a projective local measurement 
in the computational basis over $A_1$ can be implemented to purify the individual states 
of $A_1$, $\cdots$, $A_n$ all at once. Indeed, through a {\it selective} measurement over $A_1$, the demon obtains
\begin{equation}
|\psi^{(i)}_{A_1 \cdots A_n} \rangle \longrightarrow \left[ 
\begin{array}{c}
|0\cdots 0\rangle \,\,  (\textrm{with  probability} \,\, |\alpha|^2 )\\ 
|1\cdots 1\rangle \,\,  (\textrm{with  probability} \,\, |\beta|^2 )
\end{array}%
\right.
\label{meas-WC}
\end{equation}
The outcome of this measurement (either $0$ or $1$) is recorded in the classical memory of the demon, whose state 
(for an outsider) is a probability distribution that can be described by the classical density operator 
$\rho_D = |\alpha|^2 |0\rangle \langle 0|  +  |\beta|^2 |1\rangle \langle 1|$. At this first step, the classical 
demon is then able to extract the  work $n \, k\, T \, \log 2$, with the energy cost of erasure of the demon's bit 
given by $\, k\, T \, S(\rho_D) = \, k\, T \, S$.  
Therefore, we obtain the that the net classical work is given by 
\begin{equation} 
W^C = n \, k\, T \, \log 2  - \, k\, T \, S .
\label{work-C-sd}
\end{equation}
Note that the work $W^C$ in Eq.~(\ref{work-C-sd}) could also be obtained by decohering the quantum demon's qubit 
in the quantum circuit of Fig.~\ref{f2} similarly as discussed for the bipartite case in Ref.~\cite{Zurek:03}. 
For the next steps ($i>1$), the classical demon is able to drive the whole system to a fully classical state by performing 
a suitable non-selective measurement over subsystems $A_j$, with $j<i$. For example, at step $i=2$, 
by non-selectively measuring $A_1$ in the computational basis, the classical demon obtains the classical probability 
distribution 
$\Phi_{A_1}(\rho^{(i)}_{A_1 \cdots A_n}) = \alpha|^2 |0 \cdots 0 \rangle \langle 0 \cdots 0|  +  |\beta|^2 |1 \cdots 1 \rangle \langle 1 \cdots 1|$, 
with $\rho^{(i)}_{A_1 \cdots A_n}  = |\psi^{(i)}_{A_1 \cdots A_n} \rangle \langle \psi^{(i)}_{A_1 \cdots A_n} |$. 
However, because $\Phi_{A_1}(\rho^{(i)}_{A_1 \cdots A_n})$ is already fully classical, the quantum demon cannot obtain any extra advantage from 
this step on by taking any other $A_i$ ($i>1$) as an apparatus. Hence, from Eqs.~(\ref{work-Q-sd}) and~(\ref{work-C-sd}),  the total quantum advantage reads
\begin{equation}
\Delta W_t = \, k\, T \,  \mathcal{D}_{th}\left( A_1 : \cdots :A_n \right) = \, k\, T \,  S,
\end{equation}
which saturates the bound in Eq.~(\ref{DN}) provided by the thermal GQD.

%%%%%%%%%%%%%%%%%%%%%%%%%%
\subsection{Tripartite Werner-GHZ mixed state}
%%%%%%%%%%%%%%%%%%%%%%%%%%

Let us show now that saturation is also possible for mixed composite states. 
In this direction, we take the a tripartite system $ABC$ described by a Werner-GHZ state, which is given by
\begin{equation}
{\rho}=\frac{\left( 1-\lambda \right)}{8} I^{\otimes 3}  
+\lambda \left\vert GHZ\right\rangle \left\langle
GHZ\right\vert ,
\end{equation}
where
$I=|0\rangle\langle 0|+|1\rangle\langle 1|$ is the $2\times 2$ identity  and
$\left\vert GHZ \right\rangle =\left( 
|000\rangle - |111\rangle \right)/\sqrt{2}$, with $0 \le \lambda \le 1$. For this state, the minimization of 
$I({\rho}) - I(\Phi_{ABC}\left( {\rho}\right))$ occurs for measurements in the $\sigma_z$ basis~\cite{Xu:12}. 
In this basis, we obtain $H(\{p_{a}\}) - S(\rho_{A}) = 0$ (and analogous expressions for $B$ and $C$). 
Therefore, from Eq.~(\ref{ThermalGQD-mi}), we can derive that the thermal GQD 
$\mathcal{D}_{th}\left( A:B:C \right)$ is equal to the ordinary GQD (see Refs.~\cite{Rulli:11,Xu:12}), 
yielding
\begin{eqnarray}
&&\mathcal{D}_{th}\left( A:B:C \right) = \left( \frac{1+ 7\lambda}{8} \right) 
\log \left( \frac{1+7\lambda}{8} \right) \nonumber \\
&& + \left( \frac{1-\lambda}{8} \right) \log \left( \frac{1-\lambda}{8} \right) 
   - 2 \left( \frac{1+3\lambda}{8} \right) \log \left( \frac{1+3\lambda}{8} \right). \nonumber
\end{eqnarray}
As in the previous example, the upper bound for work extraction given by Eq.~(\ref{DN}) also saturates, reading 
\begin{equation}
\Delta W_t = k\,T \, \mathcal{D}_{th}\left( A:B:C \right)=  k\,T \, \mathcal{D}_{th}\left( \rho | A \right).
\end{equation}
For the Werner-GHZ state, the contributions $\mathcal{D}_{th}\left( \Phi_{A}(\rho) | B \right)$ and 
$\mathcal{D}_{th}\left( \Phi_{AB}(\rho) | C \right)$ in Eq.~(\ref{tripartite-bound}) vanish by minimizing $\Phi$ 
with measurements in the $\sigma_z$ basis, since $\Phi_{A}(\rho)$ and $\Phi_{AB}(\rho)$ are fully classical states. 

%%%%%%%%%%%%%%%%%%%%%%%%%%
\subsection{Tripartite W-GHZ mixed state}
%%%%%%%%%%%%%%%%%%%%%%%%%%

We can also consider an example of a mixed state for which saturation is not achieved. To illustrate this, 
let us consider a tripartite system $ABC$ described by the W-GHZ state
\begin{equation}
{\rho}=\lambda \left\vert W\right\rangle \left\langle W\right\vert
+\left( 1-\lambda \right) \left\vert GHZ\right\rangle \left\langle
GHZ\right\vert ,
\end{equation}
where
$\left\vert W  \right\rangle =\left( \left\vert 001 \right\rangle + 
\left\vert 010 \right\rangle +\left\vert 100 \right\rangle \right)/\sqrt{3}$ and
$\left\vert GHZ \right\rangle =\left( 
|000\rangle - |111\rangle \right)/\sqrt{2}$, with $0 \le \lambda \le 1$.
\begin{figure}[th]
%%v2: Revised Figure
\centering {\includegraphics[angle=0,scale=0.35]{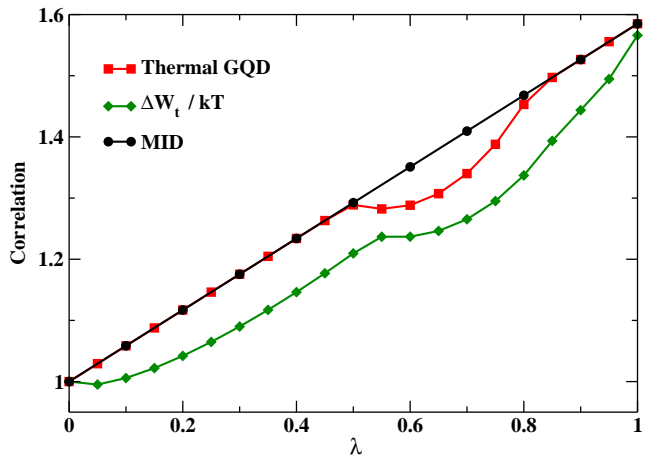}}
\caption{(Color online) Multipartite MID (black circles), thermal GQD (red squares), and the total quantum advantage $\Delta W_t$ for work extraction (green diamonds),  
as a function of $\lambda$ for the W-GHZ state. Note that MID is a global upper bound for the other quantities, while the thermal GQD is an upper bound 
for $\Delta W_t$.}
\label{f3}
\end{figure}
Note that, for $\lambda =0$ and $\lambda =1$, we have pure states, given by $|GHZ\rangle$ and 
$|W\rangle$ states, respectively. Therefore, by adopting projective measurements, we will have 
that $\mathcal{D}_{th}\left( A: B : C \right) = 1$
for $\lambda=0$. As $\lambda$ increases, we numerically find out a monotonic increase of the thermal GQD until 
$\mathcal{D}_{th}\left( A : B: C \right) = \log 3 $ for $\lambda=1$. This can be seen as a consequence of 
the absence of Schmidt decomposition for the $W$ state, which leaves GQD unconstrained by the entropy of an individual subsystem. 
For the complete range of $\lambda$, we plot the thermal GQD in Fig.~\ref{f3} as well as the total quantum advantage 
$\Delta W_t = \min_i (W^Q - W^{C_i})$. As it can be seen, the 
bound provided by the thermal GQD does not saturate for the W-GHZ state, but it is considerably tight for all the values of $\lambda$. Moreover, 
we also plot the less restrict bound provided by the multipartite MID $\mathcal{D}_{MID}\left( A : B: C \right)$, which is given by a smooth function of $\lambda$ due to its 
independence of basis optimization. Note that $\mathcal{D}_{MID}\left( A : B: C \right)$ provides a global upper bound for both $\mathcal{D}_{th}\left( A : B: C \right)$  and 
$\Delta W_t$.

%%%%%%%%%%%%%%%%%%%%%%%%%
\section{Conclusions}
%%%%%%%%%%%%%%%%%%%%%%%%%

In conclusion, we have analyzed the extraction of thermodynamic work by a Maxwell's demon in 
a multipartite quantum correlated system. In this direction, we have introduced the thermal 
GQD as a measure of quantum correlation in a multipartite scenario. Moreover, we have shown 
that this measure can be applied as an upper bound for the advantage of the quantum demon over 
its classical counterpart in a protocol of work extraction based on sequential local measurements 
over $n$ copies of the multipartite state. This result provides therefore a thermodynamic 
interpretation of the thermal GQD, which can be explored in the context of quantum thermal 
machines (see, e.g., Ref.~\cite{Park:13}). In particular, heat engines driven by the thermal GQD 
may be investigated as a resource in quantum thermodynamics. 
%%v2: 
In this scenario, it would also be interesting to investigate the quantum advantage of a Maxwell's 
demon when only a single copy of a multipartite system is available. Moreover, a further relevant topic 
is the establishment of the conditions for which the proposed thermodynamic protocol may 
constitute the optimal strategy. 
%%%
We leave these points for future research.

%%%%%%%%%%%%%%%%%%%%%%%%%%%%%%
\begin{acknowledgments}
%%%%%%%%%%%%%%%%%%%%%%%%%%%%%%
 
This work is supported by CNPq, CAPES, FAPERJ, and the Brazilian National 
Institute for Science and Technology of Quantum Information (INCT-IQ).

\end{acknowledgments}

%%%%%%%%%%%%%%%%%%%%%%%%%%%%%

\end{document}